\documentclass[10pt]{eptcs}
\usepackage{syndmeta}
\usepackage{proof}
\usepackage{amsmath}
\usepackage{amssymb}
\usepackage{amsthm}
\usepackage{stmaryrd}
\usepackage{url}
\usepackage{txfonts}
\usepackage{verbatim}
\usepackage[T1]{fontenc}
\usepackage{microtype}
\DisableLigatures{encoding=*,family=*}
\newtheorem{theorem}{Theorem}
\newtheorem{proposition}[theorem]{Proposition}
\newtheorem{lemma}{Lemma}
\newtheorem{definition}{Definition}

\newcommand{\comb}{\otimes}
\newcommand{\unobservable}{\cunit}

\newcommand{\re}{r@b.uk}
\newcommand{\ra}{starr.uk}

\title{A Veri{f}{i}ed Algebra for Linked Data}
\author{
 Ross Horne
 and
 Vladimiro Sassone
\institute{
 Electronics and Computer Science,
 University of Southampton, United Kingdom
 \email{\{rjh06r,vs\}@ecs.soton.ac.uk}
}
}
\def\authorrunning{ R. Horne \& V. Sassone }
\def\titlerunning{Linked Data Algebra}

\begin{document}

\maketitle

\begin{abstract}
A foundation is investigated for the application of loosely structured data on the Web. This area is often referred to as Linked Data, due to the use of URIs in data to establish links. This work focuses on emerging W3C standards which specify query languages for Linked Data. The approach is to provide an abstract syntax to capture Linked Data structures and queries, which are then internalised in a process calculus. An operational semantics for the calculus speci{f}{i}es how queries, data and processes interact. A labelled transition system is shown to be sound with respect to the operational semantics. Bisimulation over the labelled transition system is used to verify an algebra over queries. The derived algebra is a contribution to the application domain. For instance, the algebra may be used to rewrite a query to optimise its distribution across a cluster of servers. The framework used to provide the operational semantics is powerful enough to model related calculi for the Web.
\end{abstract}

\section{Introduction}

The application of interest is a powerful emerging idea commonly referred to as the Web of Data~\cite{Bizer2009b}. The Web of Data marks a shift from publishing documents to publishing data. The Web is based on documents which contain links to other documents. The Web of Data is concerned with resources more general than documents. Data on the Web contains links to resources described in multiple data sources. In both the case of the Web and the Web of Data the links between documents and resources, respectively, are established by a standardised global naming system --- the URI. On the Web, URIs allow documents in distributed locations with distinct ownership to refer to each other. Similarly, in a Web of Data, URIs allow data in distributed locations with distinct ownership to refer to common resources.

Suppose that the URIs are not used as a standard naming system. In this case, each data source uses its own naming system. Typically, in this case each data source is disjoint, hence traditional database techniques may be applied. This is referred to as closed world system, since the boundaries of the data source are known. For instance, classical negation can be used to determine whether some data does not appear in a data source, and schemata can constrain the structure of data. 

In contrast, the presence of URIs as a global naming system, enables an open world system. In an open world system a variety of protocols can be used to obtain data from multiple sources based on the URIs which appear. For instance, a request may be sent to a URI to directly obtain some data about that URI. Alternatively, services may be used to {f}{i}nd data relevant to a URI. In this open world setting, there is no guarantee that mechanisms {f}{i}nd all relevant data. There may always be data not known locally which refers to a resource; hence in general optimal query results cannot be obtained and classical negation cannot be applied. Another restriction in an open world system is that schemata which constrain data cannot be enforced globally.

A light semi-structured data format must be agreed for the Web of Data. The W3C recommends the Resource Description Framework (RDF) as a general format for presenting data~\cite{Carroll2004}. RDF is based on triples which consist of a subject, predicate and object. The subject, predicate and object are all named by URIs. Each URI in a triple may represent resources in different locations, hence a triple links locations. Other semi-structured data formats contain URIs, such as feeds. RDF is intended as a minimal data format to which other formats can be lifted.

Assuming that Linked Data can be gathered, observations about Linked Data can be made. The W3C recommendation is to use SPARQL Queries to make such observations~\cite{Seaborne2008}. In this work, to model this scenario, both RDF Data and SPARQL Queries are internalised in a process calculus. The operational semantics of the process calculus speci{f}{i}es how queries and data interact, to realise the W3C recommendations. The operational semantics are realistic since there is no guarantee of maximal responses, only that responses are correct.

Two SPARQL Queries may be indistinguishable with respect to their operational behaviour. Such operationally equivalent queries are bisimilar. In this work, bisimulation is used to derive an algebra over SPARQL Queries. The algebra agrees with expected equivalences analogous to those uncovered by relational algebra and exposes some new equivalences. The derived algebra can be used to rewrite a query to a normal form. Normal forms are useful for optimisation purposes. A query can be optimised before being distributed over multiple data sources. Distribution of queries is a key challenge for enabling a Web of Data~\cite{Bizer2009}.

Section~\ref{section:syntax} presents a syntax and semantics for RDF triples, SPARQL queries and processes which internalise both triples and queries.
Section~\ref{section:lts} provides an alternative operational semantics using a labelled transition system. The labelled transition system is proven to be sound with respect to the reduction system.
Section~\ref{section:bisimulation} introduces two notions of equivalence over the calculus, which correspond to the two operational semantics. Bisimulation for the labelled transition system is proven to be complete with respect to contextual equivalence for the reduction system. An algebra for queries is veri{f}{i}ed using bisimulation.

\section{A syntax and semantics for the syndication calculus}
\label{section:syntax}

The concrete syntax for both RDF and SPARQL Query are speci{f}{i}ed in W3C recommendations~\cite{Carroll2004,Seaborne2008}. Here an abstract syntax is presented to model the core features of the concrete syntax. This abstract syntax is easier to de{f}{i}ne than the concrete syntax, which is sugared to make programming easier.

The operational semantics of the calculus is speci{f}{i}ed as a reduction system. The syntax and rules of the reduction system borrow from a fragment of Linear Logic, extended with a continuation. Related work has investigated other approaches to using Linear Logic for both query languages and process calculi~\cite{Kobayashi93,Miller1994,Scott1994}.

Note that the description of the syntax and reduction system is brief. A similar syntax and reduction system are extensively discussed in the thesis of the {f}{i}rst author~\cite{Horne2011}. The main contribution of this paper is the bisimulation results for queries.

\subsection{A syntax for RDF triples}

An abstract syntax for triples conveys the RDF data format. The atoms of the syntax are names and literals. Names represent occurrences of URIs, which are represented by identi{f}{i}ers in italics, such as $\textit{John}$ or $\textit{knows}$. Literals are basic data values, such as the strings \mstring{Paul} or \mstring{77-3426}. The de{f}{i}nition of literals in the XML Schema Datatypes speci{f}{i}cation~\cite{Biron2004} is assumed. Variables $a ,b\hdots$ and $x, y\hdots$ represent place holders for names and literals respectively.

\begin{comment}
\begin{figure}
\begin{gather*}
\begin{array}{rlr}
a \quad \mbox{name}
\qquad
x \quad \mbox{literal}
\qquad
o \Coloneqq & a \mid x \quad      &\mbox{name or literal}
\end{array}
\qquad\qquad
C \Coloneqq \triple{a}{a}{o} \quad \mbox{triple}
\end{gather*}
\caption{The syntax of triples.}
\label{figure:syntax-content}
\end{figure}
\end{comment}

A triple consists of three components: the subject, the predicate and the object, which is written $\triple{\textit{subject}}{\textit{predicate}}{\textit{object}}$. The subject is related by the predicate to the object, similarly to simple sentences in English of form subject-verb-object, where URIs and literals are used instead of words. The syntax ensures that literals can only appear as the object of a triple. The example below presents two RDF triples.
\[
 \triple{b_4}{\textit{home}}{\textit{\ra}}
\qquad\qquad
 \triple{b_4}{\textit{give\_name}}{\mstring{Ringo}}
\]

Predicates are names such as $\textit{home}$. For instance, the {f}{i}rst triple above means that a subject $b_4$ is related by predicate $\textit{home}$ to object $\textit{\ra}$. The second triple above indicates that subject $b_4$ is related by predicate $\textit{given\_name}$ to the literal \mstring{Ringo}.

\begin{comment}
\qquad
\triple{\textit{\ra}}{\textit{type}}{\textit{of{f}{i}cial}}

Triples can use the RDF type predicate, represented above by the name \textit{type}. This can be misleading as the object of this predicate is not generally used as a conventional type. The object of a triple with predicate \textit{type} resides at the level of terms and is treated as any other name. If an application requires that some URIs are used as conventional types, then a light type system may be adopted~\cite{Sassone2011}. This work uses triples as only terms rather than types, in line with most applications.
\end{comment}

\subsection{A syntax for SPARQL queries}
\label{section:processes}

In this section an abstract syntax for queries, Fig.~\ref{figure:syntax-processes}, represents the core features of SPARQL Query~\cite{Seaborne2008}. SPARQL Queries are used to read from RDF triples. Synchronisation constructs allow substantial queries to be expressed. The syntax of processes, also in Fig.~\ref{figure:syntax-processes}, demonstrates how both queries and content can be internalised in a process calculus, which suggests a high level language for Linked Data, which uses query results. In this model, persistently stored triples are used to answer queries. A stored triple is indicated by an underscore.

\begin{figure}
\begin{gather*}
\begin{array}{rlr}
\phi ::= & \cunit            & \mbox{true} \\
          \mid & \zero             & \mbox{false} \\
          \mid & \phi \vee \phi    & \mbox{or} \\
          \mid & \phi \wedge \phi  & \mbox{and} \\
          \mid & \neg \phi         & \mbox{not} \\
          \mid & \hdots            & \mbox{etc.}
\end{array}
\qquad\qquad
\begin{array}{rlr}
U ::= & \ask{C}                               & \mbox{asked triple} \\
       \mid & \filter{\phi}                         & \mbox{{f}{i}lter} \\
       \mid & U \oplus U                            & \mbox{choice} \\
       \mid & U \tensor U                           & \mbox{tensor} \\
       \mid & \cselect{a} U                         & \mbox{select name} \\
       \mid & \cselect{x} U                         & \mbox{select literal} \\
       \mid & \exponential U                        & \mbox{iteration} \\
       \mid & U \then P                             & \mbox{then}
\end{array}
\qquad\qquad
\begin{array}{rlr}
P ::= & \bot            & \mbox{nothing} \\
       \mid & P \cpar P       & \mbox{par} \\
       \mid & \cscope{a} P    & \mbox{blank node} \\
       \mid & U               & \mbox{query} \\
       \mid & \ins{C}         & \mbox{stored triple} \\
\end{array}
\end{gather*}
\caption{The syntax of constraints ($\phi$), queries ($U$) and processes ($P$), over triples ($C$).}
\label{figure:syntax-processes}
\end{figure}

\paragraph{Ask queries and multiplicative operators.}

The simplest `ask' query provides a triple to be matched. There are three multiplicative operators: a tensor product ($\tensor$) for synchronously joining queries, a par operator ($\cpar$) for composing processes in parallel and the operator then ($\then$) for guarding a process with a query. The difference between tensor and par is that queries composed using tensor must happen simultaneously (in the same atomic step), whereas processes composed in parallel may be used in different atomic steps. Tensor is the implicit join of queries used in SPARQL. Then and par are part of a higher level language, where query results are immediately used. These operators are multiplicative since they control the sharing of resources.

\begin{comment}
The following demonstrates the use of the three multiplicatives and basic queries. Two queries are composed in parallel. The first query asks for two triples to be matched simultaneously. The second query ask for one triple to be answered and then another triple to be answered. Here all continuation processes can be triggered.
\[
\begin{array}{l}
 \left(
  \left(
   \ask{
    \triple{b_4}{\textit{name}}{\mstring{Ringo}}
   }
   \then P
  \right)
  \tensor
  \left(
   \ask{
    \triple{b_4}{\textit{home}}{\textit{\ra}}
   }
   \then Q
  \right)
 \right)
 \cpar
\\
\qquad\qquad\qquad
\left(
 \ask{
  \triple{b_4}{\textit{name}}{\mstring{Ringo}}
 }
 \then
 \left(
  \ask{
   \triple{b_4}{\textit{home}}{\textit{\ra}}
  }
  \then
  R
 \right)
\right)
\end{array}
\]
\end{comment}

\paragraph{The additive operators and select queries.}

There are three additive operators: choose ($\oplus$), select ($\bigvee$) and the blank node quanti{f}{i}er ($\bigwedge$). The choose operator presents a choice between two queries, hence models the SPARQL keyword $\texttt{UNION}$.
The select operator is a quanti{f}{i}er which binds a variable. Select is used to model \texttt{SELECT} queries in SPARQL, which discover names and literals. The names and literals discovered can also be bound in a continuation process, hence value passing is modelled at a high level.
Blank node quanti{f}{i}ers provide a model for blank nodes in RDF~\cite{Carroll2004}. A blank node is a local name where the scope of the blank node is indicated by the scope of the quanti{f}{i}er. Blank nodes allow further data structures to be represented in RDF, including XML.

\begin{comment}
In the following example only one continuation can be triggered, depending on the guard which is answered.
\[
\left(
 \ask{
  \triple{b_4}{\textit{home}}{\ra}
 }
 \then P
\right)
\oplus
\left(
 \ask{
  \triple{b_4}{\textit{email}}{\re}
 }
 \then Q
\right)
\]
\end{comment}

\begin{comment}
The following example extends the previous example by selecting three names.
%
\[
\cselect{a}
\left(
\cselect{w}
\left(
\cguard{
 \triple{a}{\textit{home}}{w}
}P
\right)
\oplus
\cselect{e}
\left(
\cguard{
 \triple{a}{\textit{email}}{e}
}Q
\right)
\right)
\]
%
In the above example $a$ and $w$ are bound in the {f}{i}rst continuation and $a$ and $e$ are bound in the second continuation. Only one continuation is chosen, depending on the triple that can be matched.
\end{comment}

\begin{comment}
The example below demonstrates two triples which share a blank node. The name of the blank node is quantified to indicate the scope in which it can be used.
\[
\cscope{b_4}
\left(
 \begin{array}{l}
 \cres{
  \triple{\textit{song}}{\textit{lyricist}}{b_4}
 }
 \cpar
 \cres{
  \triple{b_4}{\textit{name}}{\mstring{Ringo}}
 }
 \end{array}
\right)
\]
\end{comment}

\paragraph{Constraints and optional queries.}

A constraint may be used in a query. Constraints form a Boolean algebra of basic predicates, such as inequalities and regular expressions. The speci{f}{i}cation of constraints can be found under the keyword \texttt{FILTER} in the recommendation~\cite{Seaborne2008}. A choice between a query and true models an optional query in SPARQL, so the keyword $\texttt{OPTIONAL}$ is de{f}{i}ned as follows: $\texttt{OPTIONAL}\,U \triangleq U \oplus \cunit$.

\begin{comment}
In the following example a variable is constrained by a regular expression.
\[
\cselect{a}
\cselect{x}
\left(
 \left(
  \ask{
   \triple{a}{\textit{name}}{x}
  }
  \then
  P
 \right)
 \tensor
 \filter{
  \left(x \in \mstring{P.*}\right)
 }
 \tensor
 \left(
 \cselect{y}
 \left(
  \ask{
   \triple{a}{\textit{phone}}{y}
  }
  \then Q
 \right)
 \oplus
 \cunit
 \right)
 \right)
\]
The above query also offers a choice between part of a the query and true ($\cunit$). This indicates that the second part of the query is optional.
\end{comment}

\paragraph{Repeated queries and iteration.}

A common requirement of a query language is that more than one result can be obtained. Bounded multiple copies of queries can be synchronously posed, using queries with natural number exponents and {f}{i}nite sums. Exponents and sums are just abbreviations de{f}{i}ned as follows.
\[
U^{0} \triangleq \cunit
\qquad
U^{n+1} \triangleq U \otimes U^{n}
\qquad
\Sigma_{n=0}^{0} U^{n} 
\triangleq
\cunit
\qquad
\Sigma_{n=0}^{k + 1} U^{n}
\triangleq
\Sigma_{n=0}^{k} U^{n} \oplus U^{k+1}
\]
A natural number exponent $n$ repeatedly applies the tensor product, so the query must be answered exactly $n$ times. The sum with bound $n$ allows the query to be answered between $0$ and $n$ times. Sums model the keyword \texttt{LIMIT}, such that $U\,\texttt{LIMIT}\,k \triangleq \Sigma_{n=0}^{k} U^n$.

Unbounded iteration of queries is indicated by an explicit operator ($\exponential$), which allows zero or more copies of a query to be answered. Note that iteration differs from replication in common process calculi. All copies of an iterated query must be answered simultaneously using disjoint resources.

\begin{comment}
The example below posses a query twice and an unbounded number of times, hence the query is answered at least twice. 
%
\[
\left(
\left(
\cselect{a}
 \left(
 \ask{
  \triple{a}{\textit{status}}{\textit{official}}
 }
 \then R
 \right)
\right)^2
\otimes
\exponential
\cselect{a}
 \left(
 \ask{
  \triple{a}{\textit{status}}{\textit{official}}
 }
 \then R
 \right)
\right)
\]
%
\end{comment}

%\input{related}

\begin{comment}
\paragraph{Note on the concrete syntax.}

For the reader interested in the W3C concrete syntax (an ASCII syntax sugared to look like common query languages), connections are briefly summarised. Constraints are embedded in SPARQL using the key word \texttt{FILTER}. A choice between two queries is represented by the key word \texttt{UNION}. The tensor product of queries is implied by juxtaposition. The select quantifiers model the \texttt{SELECT} key word, which can be nested using sub-queries in the latest SPARQL drafts. SPARQL results are omitted since the continuation process is immediately provided. The $\texttt{OPT}$ key word is a secondary concept captured by the choice and the true filter, as follows $U \oplus \cunit$. Iteration is implicit unless limits are indicated. For instance, $\texttt{U LIMIT 10}$ can corresponds to $\Sigma_{n=0}^{10} U^n$. This covers the core features of SPARQL query.
\end{comment}

\subsection{A reduction system for the calculus}
\label{section:operational-semantics}

The reduction system presents a concise operational semantics for the calculus. The reduction system is de{f}{i}ned by a structural congruence and a relation over processes called the commitment relation. A further preorder over triples formalises key features of RDF Schema (RDFS~\cite{Brickley2004}). RDFS is a light extension to RDF, which improves interoperability by resolving aliases between URIs.

The structural congruence ($\equiv$ in Fig.~\ref{figure:structure}) is de{f}{i}ned such that $(P, \cpar, \cend)$ forms a commutative monoid. Alpha conversion can also be applied to blank node quanti{f}{i}ers. Furthermore, blank node quanti{f}{i}ers can be eliminated in the presence of nothing, commute and distribute over par.  All reductions are considered up to structural congruence --- as standard in process calculi.
%
\begin{comment}
The structural congruence can also allows blank node quanti{f}{i}ers to change scope. Blank node quanti{f}{i}ers can be eliminated in the presence of nothing, commute and distribute over par.  
\end{comment}

\begin{figure}
\begin{gather*}
P \cpar \bot \equiv P
\qquad
P \cpar Q \equiv Q \cpar P
\qquad
P \cpar (Q \cpar R) \equiv (P \cpar Q) \cpar R
\\[12pt]
\cscope{a} \cend \equiv \cend
\qquad
\cscope{a} \cscope{b} P \equiv \cscope{b} \cscope{a} P
\qquad
\cscope{a} P \cpar Q \equiv \cscope{a} (P \cpar Q) \quad a \not \in \fn{P}
\end{gather*}
\caption{The structural congruence over processes.}
\label{figure:structure}
\end{figure}

The commitment relation ($\Trans$ in Fig.~\ref{figure:transitions}) speci{f}{i}es atomic operational steps. The process on the left of the commitment relation, becomes the process on the right. A commitment is performed atomically.

\begin{figure}
\begin{gather*}
\infer{
 \ins{C} \cpar \ask{D}
 \Trans
 \ins{C}
}{
 C \sqsubseteq D
}
\qquad
\infer{
 \filter{\phi} \Trans \cend
}{
 \vDash \phi
}
\qquad
\infer{
 P \cpar \left(U \oplus V\right) \Trans Q
}{
 P \cpar U \Trans Q
}
\qquad
\infer{
 P \cpar \left(U \oplus V\right) \Trans Q
}{
 P \cpar V \Trans Q
}
\qquad
\infer{
 P \cpar Q \cpar \left(U \tensor V\right)
 \Trans
 P' \cpar Q'
}{
 P \cpar U \Trans P'
 &
 Q \cpar V \Trans Q'
}
\\[12pt]
\exponential U \Trans \cend
\qquad
\infer{
 P \cpar \exponential U \Trans Q
}{
 P \cpar U \Trans Q
}
\qquad
\infer{
 P \cpar \exponential U
 \Trans Q
}{
 P \cpar \left(\exponential U \tensor \exponential U\right)
 \Trans Q
}
\qquad
\infer{
 P \cpar \cselect{a} U \Trans Q
}{
 P \cpar U\sub{a}{b} \Trans Q
}
\qquad
\infer{
 P \cpar \cselect{x} U \Trans Q
}{
 P \cpar U\sub{x}{v} \Trans Q
}
\\[12pt]
\infer{
 P \cpar \left(U \then R\right) \Trans Q \cpar R
}{
 P \cpar U \Trans Q 
}
\qquad
\infer{
 P \cpar Q \Trans P' \cpar Q
}{
 P \Trans P'
}
\qquad
\infer[
 a \not \in
 \fn{
  \begin{array}{l}
  P, P', \beta
  \end{array}
 }
]{
 P \cpar \cscope{a}Q \Trans P' \cpar \cscope{a}Q'
}{
 P \cpar Q \Trans P' \cpar Q'
}
\end{gather*}
\caption{
 Commitment rules: ask, {f}{}{i}lter, choose left, choose right, tensor, weakening, dereliction, contraction, select name, select literal, guard, context, and blank node ($\mathrm{fn}$ indicates the free names).
}
\label{figure:transitions}
\end{figure}

Working with aliases for URIs is a key problem in Linked Data~\cite{Alani2002}. Aliases arise since different data sources use different URIs for similar purposes. For instance, in the context of a song, predicate $\textit{lyricist}$ may be more speci{f}{i}c than predicate $\textit{creator}$ (see subPropertyOf in RDFS~\cite{Brickley2004}). Similarly, $\textit{song}_0$ and $\textit{song}_1$ may be URIs for the same song (see sameAs in OWL~\cite{Alani2002}). Hence the aliases $\textit{lyricist} \sqsubseteq \textit{creator}$ and $\textit{song}_0 \sqsubseteq \textit{song}_1$ may be assumed. The application speci{f}{i}c set of alias assumptions is referred to as $\beta$. The transitive reflexive closure of $\beta$ gives rise to a preorder ($\sqsubseteq$) over URIs.

\paragraph{The ask axiom, guard rule and alias assumptions.}

\begin{comment}
\begin{figure}
\begin{gather*}
\infer{
 a \sqsubseteq b
}{
 a \sqsubseteq b \in \beta
}
\qquad
\infer{
 \triple{a}{p}{c} \sqsubseteq \triple{b}{q}{d}
}{
 a \sqsubseteq b
 &
 p \sqsubseteq q
 &
 c \sqsubseteq d
}
\qquad
\infer{
 \triple{a}{p}{v} \sqsubseteq \triple{b}{q}{v}
}{
 a \sqsubseteq b
 &
 p \sqsubseteq q
}
\end{gather*}
\caption{The preorder over triples: reflexivity, transitivity and refine triple. Alias assumption are drawn from the reflexive transitive closure of a relation over names -- denoted $\beta$ throughout.}
\label{figure:preorder}
\end{figure}
\end{comment}

The following example demonstrates the interaction of an ask query with a continuation and a stored triple.
The axiom `ask' allows a query triple and a stored triple to interact. The stored triple remains available after the commitment. The axiom `guard' makes the continuation process available after the commitment. 
\[
 \cres{
  \triple{\textit{song}_0}{\textit{lyricist}}{b_4}
 }
 \cpar
 \left(
  \cguard{
   \triple{\textit{song}_1}{\textit{creator}}{b_4}
  } P
 \right)
\Trans
 \cres{
  \triple{\textit{song}_0}{\textit{lyricist}}{b_4}
 }
 \cpar
 P
\]
\noindent

Above, the conditions for a match are relaxed by the preorder over triples ($\sqsubseteq$). The preorder over triples is the point-wise extension of the preorder over URIs introduced above.
%The alias assumptions are an application dependant design decision. 
%The preorder over names allows a general query to be answered using more specific data.

%The second alias above models () same as in OWL. These preorders need not be distinguished, since resources and predicates are just URIs.

\paragraph{The tensor and select rules.}

The following example demonstrates two synchronised queries, in the presence of two stored triples. The {f}{i}rst query poses a pattern to match, while the second query selects a name with respect to a pattern.
\[
\begin{array}{l}
 \left(
  \ask{
   \triple{b_2}{\textit{role}}{\textit{singer}}
  }
  \tensor
  \cselect{b}
  \left(
  \ask{
   \triple{b}{\textit{role}}{\textit{guitarist}}
  }
  \then P
  \right)
 \right)
 \cpar
\\
 \cres{
  \triple{b_2}{\textit{role}}{\textit{singer}}
 }
 \cpar
 \cres{
  \triple{b_3}{\textit{role}}{\textit{guitarist}}
 }
\end{array}
\Trans
\begin{array}{l}
 \cres{
  \triple{b_2}{\textit{role}}{\textit{singer}}
 }
 \cpar
 \\
 \cres{
  \triple{b_3}{\textit{role}}{\textit{guitarist}}
 }
 \cpar
 P \sub{b}{b_3}
\end{array}
\]
\noindent
In the above example, the `tensor' rule divides the stored triples between the two parts of the query. On the left the `select' rule is applied. The `select' rule substitutes a suitable URI for the quanti{f}{i}ed name. The result is that a URI is passed to the continuation.

\paragraph{The choose rule.}

The following example demonstrates a choice between queries. The `choose left' rule is used in this case.
%The `select' rule is then used to input a name, which appears in the continuation process.
\[
\cselect{a}
\left(
 \left(
 \ask{
  \triple{a}{\textit{knows}}{b_2}
 }
 \then P
 \right)
 \oplus
 \left(
 \ask{
  \triple{b_2}{\textit{knows}}{a}
 }
 \then Q
 \right)
\right) 
 \cpar
 \cres{
  \triple{b_1}{\textit{knows}}{b_2}
 }
\Trans
 \cres{
  \triple{b_1}{\textit{knows}}{b_2}
 }
 \cpar
 P\sub{a}{b_1}
\]
The query result determines the continuation triggered.

\paragraph{Constraints in queries.}

The example query below selects a literal. The data literal appears in a triple and a constraint. The rules ensure that both a suitable triple appears and the constraint imposed holds.
\[
 \cselect{x}
 \left(
 \filter{
  \left(|x| \leq 5\right)
 }
 \otimes
 \ask{
  \triple{b_1}{\textit{name}}{x}
 }
 \then P
 \right)
 \cpar 
 \cres{
  \triple{b_1}{\textit{name}}{\mstring{John}}
 }
\Trans
 \cres{
  \triple{b_1}{\textit{name}}{\mstring{John}}
 }
 \cpar
 P\sub{x}{\mstring{John}}
\]
\noindent The satisfaction relation for evaluating constraints $\vDash$, is left to the W3C recommendation~\cite{Seaborne2008}. Satisfaction is assumed to de{f}{i}ne a Boolean algebra of constraints.

\paragraph{The rules for iteration of queries.}

The example below demonstrates iteration used to answer two copies of the same query. Two iterated queries are answered using `dereliction', which are combined using the conventional tensor rule. The `contraction' rule then reduces the combined queries to a single query.
\[
 \exponential
  \cselect{c}
  \left(
   \ask{
    \triple{c}{\textit{is}}{\textit{busy}}
   }
   \then P
  \right)
 \cpar
 \cres{
  \triple{b_2}{\textit{is}}{\textit{busy}}
 }
 \cpar
 \cres{
  \triple{b_3}{\textit{is}}{\textit{busy}}
 }
\Trans
 \cres{
  \triple{b_2}{\textit{is}}{\textit{busy}}
 }
 \cpar
 \cres{
  \triple{b_3}{\textit{is}}{\textit{busy}}
 }
 \cpar
 P\sub{c}{b_2}
 \cpar
 P\sub{c}{b_3}
\]
\noindent
A continuation for each result is triggered. Note the `weakening' rule could be used to allow the query to be answered zero times.

\paragraph{Blank nodes as quanti{f}{i}ers.}

The example below demonstrates a query which discovers a blank node. The `blank node' rule uses a temporary name to represent the blank node. The result is that the scope of the blank node quanti{f}{i}er is extended to include the continuation, which receives the blank node.
\[
\begin{array}{l}
\cselect{c}
\left(
\cguard{
 \triple{c}{\textit{creator}}{b_2}
} U
\right)
\cpar
\\
\cscope{a}
\left(
 \cres{
  \triple{a}{\textit{author}}{b_2}
 }
 \cpar
 \cres{
  \triple{a}{\textit{status}}{\textit{open}}
 }
\right)
\end{array}
\Trans
\cscope{a}
\left(
 \begin{array}{l}
 U\sub{c}{a}
 \cpar
 \\
 \cres{\triple{a}{\textit{author}}{b_2}}
 \cpar
 \cres{
  \triple{a}{\textit{status}}{\textit{open}}
 }
 \end{array}
\right)
\]
The alias $\textit{author} \sqsubseteq \textit{creator}$ is assumed above. The temporary name must not appear in the alias assumptions ($\beta$). The unused stored triple is idled.

%The reduction system provides an operational semantics for SPARQL queries in the presence of RDF content.
Rules for an additive disjunction, tensor product, existential quanti{f}{i}cation, universal quanti{f}{i}cation and iteration, are borrowed from Linear Logic~\cite{Girard1987}. The sequent calculus is extended to indicate a continuation process, constraints extend the basic units with a Boolean algebra, and a preorder accommodates aliases over names.

\section{A labelled transition system for the operational semantics}
\label{section:lts}

The operational semantics can be expressed as a labelled transition system. This provides an alternative operational semantics to the reduction system. This alternative semantics allows the behaviour of queries and data to be evaluated separately and then composed. Lemma~\ref{lemma:labels} veri{f}{i}es that the labelled transition system and reduction system describe the same behaviour.

\subsection{The purpose of labels}

A labelled transition consists of two processes and a label. The {f}{i}rst process is the process before the transition. The label is a constraint on the context in which a transition can take place. The second process is the resulting process after the transition.

\begin{comment}
\begin{figure}
\begin{gather*}
\begin{array}{rlr}
 E ::= & \unobservable & \mbox{unit} \\
  \mid & C             & \mbox{triple} \\
  \mid & E \comb E     & \mbox{combination}
\end{array}
\qquad\qquad
\begin{array}{c}
\left(E \comb F\right) \comb G \equiv \left(E \comb F\right) \comb G
\qquad
E \comb \cunit \equiv E
\\[10pt]
E \comb F \equiv F \comb E
\end{array}
\end{gather*}
\caption{
 The syntax of labels and the structural congruence over labels.
}
\label{figure:labels}
\end{figure}
\end{comment}

The labels are formed from a commutative monoid over triples $(E, \comb, \cunit)$. A label indicates the inputs and outputs of a process. An input indicates that a process can proceed if it can receive the triples on the label from its context. An output indicates that a process outputs the triple on the label to its context. For instance, the query below inputs a triple; while the stored triple below outputs a triple. 
\[
\ask{
 \triple{b_4}{\textit{knows}}{b_3}
}
\then
P
\lts{
 \triple{b_4}{\textit{knows}}{b_3}
}
P
\qquad\qquad
\ins{
 \triple{b_4}{\textit{knows}}{b_3}
}
\lts{
 \co{
  \triple{b_4}{\textit{knows}}{b_3}
 }
}
\ins{
 \triple{b_4}{\textit{knows}}{b_3}
}
\]
A relevant interpretation is that the {f}{i}rst transition above is an action from the perspective of a client which resolves a query; whereas the second is an action from the perspective of a server that provides a triple. Two processes composed in parallel with matching inputs and outputs may interact. For instance, the above processes can be composed, resulting in the following transition. The unit label indicates an operational step without side effects.
\[
\ask{
 \triple{b_4}{\textit{knows}}{b_3}
}
\then
P
\cpar
\ins{
 \triple{b_4}{\textit{knows}}{b_3}
}
\lts{
 \unobservable
}
P
\cpar
\ins{
 \triple{b_4}{\textit{knows}}{b_3}
}
\]

Output labels can also indicate extruded names. For instance, the example below extrudes the name $a$. The extruded names represent blank nodes where the scope of the blank node quanti{f}{i}er may be extended. This is similar to extrusion of new names in the $\pi$-calculus~\cite{Milner1992}.
\[
\cscope{a}
\ins{
 \triple{a}{\textit{has}}{\textit{paper}}
}
\cpar
\ins{
 \triple{b_2}{\textit{has}}{\textit{stone}}
}
\lts{
 a \mid
 \co{
  \triple{a}{\textit{has}}{\textit{paper}}
 }
}
\ins{
 \triple{a}{\textit{has}}{\textit{paper}}
}
\cpar
\ins{
 \triple{b_2}{\textit{has}}{\textit{stone}}
}
\]
The commutative monoid rules can always be applied to reorder labels.
%The final rule of the congruence cancels out matching inputs and outputs.
%The structural congruence just simplifies the rules of the labelled transition system.

\subsection{Labelled transitions for queries}

The input transitions allow the behaviour of a query to be modelled independently. The rules for queries are presented in Fig.~\ref{figure:lts-1}. The rules accumulate RDF triples on an input label, which represents contexts in which a query may be answered.

\begin{figure}
\begin{gather*}
\infer{
 \ask{D} \inlts{C} \cend
}{
 C \sqsubseteq D
}
\qquad
\infer{
 U \then P \lts{E} Q \cpar P
}{
 U \lts{E} Q
}
\qquad
\infer{
 U \tensor V \inlts{E \comb F} P \cpar Q
}{
 U \inlts{E} P
 &
 V \inlts{F} Q
}
\qquad
\infer{
 U \oplus V \inlts{E} P
}{
 U \inlts{E} P
}
\qquad
\infer{
 U \oplus V \inlts{E} Q
}{
 V \inlts{E} Q
}
\\[10pt]
\infer{
 \filter{\phi} \inlts{\unobservable} \cend
}{
 \vDash \phi
}
\qquad
\infer{
 \cselect{a} U \inlts{E} Q
}{
 U\sub{a}{b} \inlts{E} Q
}
\qquad
\infer{
 \cselect{x} U \inlts{E} Q
}{
 U\sub{x}{v} \inlts{E} Q
}
\qquad
\exponential U \lts{\unobservable} \cend
\qquad
\infer{
 \exponential U \inlts{E} P
}{
 U \inlts{E} P
}
\qquad
\infer{
 \exponential U \inlts{E} P
}{
 \exponential U \tensor \exponential U \inlts{E} P
}
\end{gather*}
\caption{
 Labelled transitions for queries: input triple, trigger guard, tensor, choose left, choose right, {f}{i}lter, select name, select literal, weakening, dereliction and contraction.
}
\label{figure:lts-1}
\end{figure}

The `input triple' rule poses the triple as an input on the label. The triple on the label may be strengthened by the preorder over triples. The `trigger guard' rule allows a continuation process to be triggered exposing the continuation. The following example demonstrates a query consisting of a single triple and a continuation process, where the preorder $\textit{colleague} \sqsubseteq \textit{knows}$ is assumed.
\[
\ask{
 \triple{b_4}{\textit{knows}}{b_3}
}
\then
P
\lts{
 \triple{b_4}{\textit{colleague}}{b_3}
}
P
\]

Select quanti{f}{i}ers are resolved by anticipating the name or literal to input. For instance, the following labelled transition indicates that the query can be answered in a context where a name is chosen. The same name is passed to the continuation process.
\[
\cselect{a}
\left(
\ask{
 \triple{b_4}{\textit{knows}}{a}
}
\then
P
\right)
\lts{
 \triple{b_4}{\textit{knows}}{b_3}
}
P \sub{a}{b_3}
\]

Choices are resolved by anticipating the left or right branch. For instance, the following transition indicates the label and continuation which results from choosing the left branch.
\[
\left(
\ask{
 \triple{b_4}{\textit{knows}}{b_2}
}
\then P
\right)
\oplus
\left(
\ask{
 \triple{b_4}{\textit{knows}}{b_3}
}
\then Q
\right)
\lts{
 \triple{b_4}{\textit{knows}}{b_2}
}
P
\]

Tensor synchronises two queries, by composing their respective labels and continuations. For instance, the following query simultaneously inputs two triples. The continuations of both queries are triggered in parallel, with the appropriate substitutions.
\[
\begin{array}{l}
\cselect{a}
\left(
\left(
\ask{
 \triple{b_4}{\textit{knows}}{a}
}
\then
P
\right)
\join
\left(
\cselect{x}
\ask{
 \triple{a}{\textit{name}}{x}
}
\then Q
\right)
\right)
%\\\qquad\qquad\qquad\qquad\qquad\qquad
\lts{
 \triple{b_4}{\textit{knows}}{b_2}
 \comb
 \triple{b_2}{\textit{name}}{\mstring{John}}
}
P \sub{a}{b_2} \cpar Q \sub{a,x}{b_2,\mstring{John}}
\end{array}
\]

A constraint is disposed when it is satis{f}{i}ed. For instance, in the following query the length of a selected literal is constrained, but satis{f}{i}ed by the substitution.
\[
\cselect{x}
\left(
\ask{
 \triple{b_2}{\textit{name}}{x}
}
\join
\left( \left| x \right| \leq 5 \right)
\then P
\right)
\lts{
 \triple{b_2}{\textit{name}}{\mstring{John}} 
}
P \sub{x}{\mstring{John}}
\] 

Iteration anticipates the number of copies of a query to pose using weakening, dereliction and contraction. For instance, two copies of the following query are posed using contraction and dereliction. The label indicates the two separate triples which are to be answered simultaneously. Both continuations are composed in parallel.
\[
\begin{array}{l}
\exponential
\cselect{a}
\left(
\ask{
 \triple{b_4}{\textit{knows}}{a}
}
\then
P
\right)
\lts{
 \triple{b_4}{\textit{knows}}{b_2}
 \comb
 \triple{b_4}{\textit{knows}}{b_3}
}
P \sub{a}{b_2}
\cpar
P \sub{a}{b_3}
\end{array}
\]

The rules of the labelled transition system are su{ff}{i}cient to model queries.
%The labelled transitions are simpler than the transition system.

\subsection{Labelled transitions for an RDF store}

The behaviour of stored RDF triples can be modelled using output labels. The rules of output labels are presented in Fig.~\ref{figure:lts-2}. The names extruded on the label are indicated by $\alpha$, where $+$ indicates disjoint union of names. The abbreviation $\cscope{\alpha}P$ is used to indicate the quanti{f}{i}cation of all names in $\alpha$.

\begin{figure}
\begin{gather*}
\infer{
 \ins{C} \lts{\co{D}} \ins{C}
}{
 C \sqsubseteq D
}
\qquad\qquad
\infer[
 \begin{array}{l}
 a \not \in \fn{\beta}
 \end{array}
]{
 \cscope{a}P \lts{\alpha + a \mid \co{E}} Q
}{
 P \lts{\alpha \mid \co{E}} Q
}
\qquad\qquad
\infer[
 \begin{array}{l}
  a \not \in \alpha \cup \fn{E}
%  a \not \in \fn{E}
 \end{array}
]{
 \cscope{a} P
 \lts{\alpha \mid \co{E}}
 \cscope{a} Q
}{
 P \lts{\alpha \mid \co{E}} Q
}
\\[10pt]
\infer[
  \alpha \cap \fn{Q} = \emptyset
]{
 P \cpar Q \lts{\alpha \mid \co{E}} P' \cpar Q
}{
 P \lts{\alpha \mid \co{E}} P'
}
\qquad\quad
\infer[
 \begin{array}{l}
  \alpha_0 \cap \fn{Q} = \emptyset \\
  \alpha_1 \cap \fn{P} = \emptyset
 \end{array}
]{
 P \cpar Q \lts{\alpha_0 + \alpha_1 \mid \co{E \comb F}} P' \cpar Q'
}{
 P \lts{\alpha_0 \mid \co{E}} P'
 &
 Q \lts{\alpha_1 \mid \co{F}} Q'
}
\\[10pt]
\infer[
 \begin{array}{l}
  \alpha \cap (\fn{P} \cup \fn{E}) = \emptyset
 \end{array}
]{
 P \cpar Q \lts{E} \cscope{\alpha} \left(P' \cpar Q'\right)
}{
 P \lts{E \comb F} P'
 &
 Q \lts{\alpha \mid \co{F}} Q'
}
\end{gather*}
\caption{
 Process rules: output triple, open, blank node context, par context, parallel outputs and close. The symmetric versions of the par context and close rule are also assumed.
}
\label{figure:lts-2}
\end{figure}

Stored triples can output the triple on the label. The same triple appears in the continuation unchanged. The preorder over names may be used to weaken the output triple. Names are extruded on the label using the `open scope' rule. For instance, the following triple outputs a triple and extrudes the blank node, using the assumption $\textit{colleague} \sqsubseteq \textit{knows}$.
\[
\cscope{b_4}
\ins{
 \triple{b_4}{\textit{colleague}}{b_3}
}
\lts{
 b_4
 \mid
 \co{
  \triple{b_4}{\textit{knows}}{b_3}
 }
}
\ins{
 \triple{b_4}{\textit{colleague}}{b_3}
}
\]

Output labels composed in parallel can be combined. Extruded names on both labels must be disjoint to preserve the scope of blank nodes. For instance, the following transition simultaneously outputs two triples and extrudes three names.
\[
\begin{array}{l}
\cscope{b_4}
\left(
 \cscope{b_2}
 \ins{
  \triple{b_4}{\textit{knows}}{b_2}
 }
 \cpar
 \cscope{b_3}
 \ins{
  \triple{b_4}{\textit{knows}}{b_3}
 }
\right)
\\
\qquad\qquad\qquad\qquad\qquad\qquad\qquad
\lts{
 b_2, b_3, b_4
 \mid
 \co{
  \triple{b_4}{\textit{knows}}{b_2}
  \tensor
  \triple{b_4}{\textit{knows}}{b_3}
 }
}
\ins{
 \triple{b_4}{\textit{knows}}{b_2}
}
\cpar
\ins{
 \triple{b_4}{\textit{knows}}{b_3}
}
\end{array}
\]

Two parallel processes may interact using the close rule. Close allows complementary inputs and outputs to be matched. Names extruded on the output label are introduced as quanti{f}{i}ers in the continuation. Any inputs not answered remain on the resulting label, to be answered later. For instance, the following iterated query is answered twice. One copy is answered by the available process and the other copy must be answered by the context for the transition to occur. In the continuation, the scope of the blank node is extended.
\[
\begin{array}{l}
\exponential
\cselect{a}
\left(
 \triple{b_4}{\textit{knows}}{a}
 \then P
\right)
\cpar
\cscope{b_3}
\ins{
 \triple{b_4}{\textit{knows}}{b_3}
}
%\\
%\qquad\qquad\qquad\qquad\qquad\qquad
\lts{
 \triple{b_4}{\textit{knows}}{b_2}
}
\cscope{b_3}
\left(
 P \sub{a}{b_2}
 \cpar
 P \sub{a}{b_3}
 \cpar
 \ins{
  \triple{b_4}{\textit{knows}}{b_3}
 }
\right)
\end{array}
\]

The context rule for parallel composition allows a process which does not contribute to an interaction to idle. Similarly, the context rule for blank node quanti{f}{i}ers allows a blank node to be ignored in a transition if it does not appear on the label.
%All examples expressed in the reduction system also work in the labelled transition system.

\subsection{Comparison of the two operational semantics}
\label{section:elimination}

To justify the labelled transition system, the labelled transitions are compared to the reductions of the reduction system. If a unit labelled transition can be derived then the corresponding reduction can also be derived. The signi{f}{i}cance is that, given the independent perspectives of the query and the store in terms of labelled transitions, their combination satis{f}{i}es the global perspective speci{f}{i}ed by the reduction system.

%A trick used in the lemmas in this section is to consider processes up to structural congruence. Structural congruence over processes is not required for the labelled transition system, but is permitted in proofs. This is because the goal is a translation into the reduction system where structural congruence is permitted.

%Scope extrusion enables a compositional semantics. Scope extrusion is frequently used in the literature to model the ability of a process to extend a name. However, 

Scope extrusion presents technical dif{f}{i}culties. The following technical lemma reduces these dif{f}{i}culties, by eliminating scope extrusion. The proof demonstrates that combinations of opening names and closing names can be eliminated from a proof tree which uses an extruded name. % using a structural congruence.

\begin{lemma}[Elimination of extrusion]
\label{lemma:extrusion}
Suppose that a labelled transition proof uses name extrusion, but not in the conclusion. The same labelled transition, up to structural congruence, holds without any name extrusion.
\end{lemma}

Note that full proofs for all theorems are provided in the thesis of the {f}{i}rst author~\cite{Horne2011}.

Every completed labelled transition can also be expressed as a reduction, Lemma~\ref{lemma:labels}. The proof works by transforming proof trees so that labels used in interactions are eliminated.
%If the rule of the interaction of labels is cut and the reduction semantics provide cut free proof, then the following is a cut elimination result for the calculus.
%
\begin{lemma}[Elimination of labels]
\label{lemma:labels}
$P \lts{\unobservable} Q$ if and only if $P \Trans Q$.
\end{lemma}

\begin{comment}
\begin{proof}
Here only one case is considered. Given the following proof tree,
\[
\infer{
 \left( U \tensor V \right)
 \cpar
 P \cpar Q
 \lts{\unobservable}
 R \cpar S \cpar P' \cpar Q'
}{
 \infer{
  U \tensor V 
  \lts{C \comb D}
  R \cpar S
 }{
  U \lts{C} R
  &
  V \lts{D} S
 }
 &
 \infer{
  P \cpar Q \lts{\co{C \comb D}} P' \cpar Q'
 }{
  P \lts{\co{C}} P'
  &
  Q \lts{\co{D}} Q'
 }
}
\]
the following proof trees hold.
\[
\vcenter{
\infer{
 U \cpar P \lts{\cunit} R \cpar P'
}{
 U \lts{C} R
 &
 P \lts{\co{C}} P' 
}
}
\qquad
\mbox{and}
\qquad
\vcenter{
\infer{
 V \cpar Q \lts{\cunit} S \cpar Q'
}{
 V \lts{D} S
 &
 Q \lts{\co{D}} Q' 
}
}
\]
Hence, by induction, both $ U \cpar P \Trans R \cpar P'$ and $V \cpar Q \Trans S \cpar Q'$ hold. Hence $\left( U \tensor V \right) \cpar P \cpar Q \Trans R \cpar S \cpar P' \cpar Q'$ holds, by the tensor rule.
\end{proof}
\end{comment}

%This result veri{f}{i}es that the labelled transition systems for queries and for RDF stores, their combination satis{f}{i}ed the reduction semantics. 
Thus the local perspective of the labelled transition system and the global perspective of the reduction system specify the same operational capabilities.

\section{An algebra for the syndication calculus}
\label{section:bisimulation}

In this section bisimulation is introduced as the natural notion of equivalence over the labelled transition system. Bisimulation is demonstrated to be sound with respect to equivalence in the reduction system. Thus every pair of bisimilar processes are equivalent with respect to the natural notion of equivalence over the reduction system. Bisimulation is then used to verify an algebra over queries and processes.

\subsection{Bisimulation}

Processes which are capable of the same observable behaviour can be regarded as equivalent. The observable behaviour of a process is given by the labels of the labelled transition system. Observational equivalence of processes is established using the technique of (strong) bisimulation, as follows.
\begin{definition}[Bisimulation]
\label{definition:bisimulation}
Bisimulation, written $\sim$, is the greatest symmetric relation such that the following holds, for any label $l$. If $P \sim Q$ and $P \lts{l} P'$ then there exists some $Q'$ such that $Q \lts{l} Q'$ and $P' \sim Q'$.
\end{definition}

The following veri{f}{i}es that bisimulation is a congruence --- a relation which holds in any context. It is necessary that bisimulation is a congruence for it to be used as an algebra. A context is a process with a place holder for some syntax.
\begin{lemma}[Bisimulation is a congruence]
\label{lemma:context}
If $P \sim Q$ and $\mathcal{C}$ is a context, then $\mathcal{C} P \sim \mathcal{C} Q$.
%Similarly, for query contexts.
\end{lemma}

An alternative notion of equivalence is de{f}{i}ned using the reduction system. Contextual equivalence is used in related work to justify notions of bisimulation on the $\pi$-calculus and ambient calculus~\cite{Rathke2005,Merro2002}.
\begin{definition}[Contextual equivalence]
\label{definition:contextual}
Contextual equivalence, written $\simeq$, is the greatest symmetric, reduction closed, context closed relation. A relation $\mathrel{\mathcal{R}}$ is reduction closed iff $P \mathrel{\mathcal{R}} Q$ and $P \Trans P'$ then there exists some $Q'$ such that $Q \Trans Q'$ and $P' \mathrel{\mathcal{R}} Q'$. A relation $\mathrel{\mathcal{R}}$ is context closed iff $P \mathrel{\mathcal{R}} Q$ yields that $\mathcal{C} P \mathrel{\mathcal{R}} \mathcal{C} Q$, for all contexts $\mathcal{C}$.
\end{definition}

Bisimulation is sound with respect to contextual equivalence. Soundness is essential to justify the chosen notion of bisimulation. 

\begin{theorem}[Bisimulation is a contextual equivalence]
\label{theorem:bisimulation-congruence}
If $P \sim Q$ then $P \simeq Q$.
\end{theorem}
\begin{proof}
Reduction closure follows from Lemma~\ref{lemma:labels} and context closure follows from Lemma~\ref{lemma:context}.
\end{proof}

Soundness of bisimulation ensures that algebraic properties proven using bisimulation also hold for contextual equivalence. Bisimulation simpli{f}{i}es proofs in the following section. Note that completeness (contextual equivalence is a bisimulation) is not required for this work. Completeness can only be achieved in an extended version of the calculus.

\subsection{Algebraic properties of queries}
\label{section:algproperties}

Using bisimulation as an equivalence, key properties of queries are established. This section amounts to a soundness proof of the algebraic properties established. Thus if any two process are equivalent according to the algebraic properties then they are bisimilar; and furthermore, by Theorem~\ref{theorem:bisimulation-congruence}, they are contextually equivalent.

For the labelled transition system, structural congruence is not assumed, hence veri{f}{i}ed here. The proof for the distributivity of blank node quanti{f}{i}ers over par requires extensive case analysis. The case of associativity of par follows from distributivity of blank node quanti{f}{i}ers. Proofs are similar to the analogous bisimulations in the $\pi$-calculus~\cite{Milner1992}.
\begin{proposition}
%[Structural congruence]
\label{proposition:process-algebra}
The structural congruence (Fig.\ref{figure:structure}) is a bisimulation. So, $(P, \cpar, \cend)$ forms a commutative monoid. Blank node quanti{f}{i}ers annihilate with $\bot$, commute, and distribute over $\cpar$.
\end{proposition}

Bisimulation reveals some canonical algebraic properties of queries. Firstly, queries form an idempotent semiring. Semirings are ubiquitous in computer science. A notable feature of semirings is that the ideals of a semiring form a semiring.
\begin{proposition}
%[Queries as a semiring]
\label{proposition:semiring}
$(U, \tensor, \oplus, \cunit, \zero)$ is a commutative idempotent semiring. That is, $(U, \tensor, \cunit)$ is a commutative monoid, $(U, \oplus, \zero)$ is idempotent commutative monoid. $\tensor$ distributes over $\oplus$ and $\zero$ annihilates with $\tensor$.
\end{proposition}

Idempotent semirings have a natural preorder, given by $U \leq V$ iff $U \oplus V \sim V$. Hence queries have this natural preorder. An immediate consequence is that choice is a colimit, i.e. least upper bound, of two queries.
%Thus, given the principal ideals for $U$ and $V$ the principal 
\begin{proposition}
%[Choice as a colimit]
\label{proposition:colimitchoice}
Choice is a colimit of its branches. That is, $V \leq W$ and $U \leq W$, if and only if $V \oplus U \leq W$.
\end{proposition}

The preorder over queries can be used to optimise queries. If a query offers a choice between a query and a weaker query, with respect to the preorder, the stronger branch may be eliminated. For instance, in related work~\cite{Arenas2009}, is is claimed that $U \mathrel{\texttt{OPTIONAL}} (V \mathrel{\texttt{OPTIONAL}} W)$ is not the same as $(U \mathrel{\texttt{OPTIONAL}} V) \mathrel{\texttt{OPTIONAL}} W$. Under the interpretation of \texttt{OPTIONAL} in the calculus it holds that $U \tensor ((V \tensor (W \oplus \cunit)) \oplus \cunit) \leq U \tensor \left(\left(V \oplus \cunit\right) \tensor \left(W \oplus \cunit\right)\right)$, by distributivity, commutativity and idempotency. So the {f}{i}rst is a stronger query.

A single rule is su{ff}{i}cient to capture the algebra of the select quanti{f}{i}er. From this algebra common equalities can be derived. The derived rules are suitable for the optimisation technique of flattening nested selects used in relational algebra~\cite{Cyganiak2005}. The proof of commutativity of quanti{f}{i}ers requires capture avoiding substitution to be assumed. The presence of the tensor in the rule is required to prove that $\cselect{a} U \tensor V \leq \cselect{a} (U \tensor V)$, when $a \not \in \fn{V}$.

\begin{proposition}
%[Select as a colimit]
\label{proposition:colimitselect}
Selects are colimits of substitutions. So, $U\sub{a}{b} \tensor V \leq W$ for all $b$, if and only if $\cselect{a}U \tensor V \leq W$. Immediate consequences are that, select commutes, distributes over choice, is annihilated by true and distributes over tensor. Furthermore, alpha conversion of bound variables is veri{f}{i}ed.
\end{proposition}

The following rules of regular algebra hold. The {f}{i}rst of the rules is su{ff}{i}cient to demonstrate that $\exponential V \tensor U$ is a {f}{i}xed point of the (monotone) map $W \mapsto U \oplus (V \tensor W)$. The second rule demonstrates that $\exponential V \tensor U$ is the least such {f}{i}xed point. Historically, Redko demonstrated that no {f}{i}nite collection of equations could axiomatise iteration~\cite{Redko1964}. The formulation below, was proven to be complete by Kozen~\cite{Kozen1994}.

\begin{proposition}
\label{proposition:exponential}
An iterated query expands as follows $\exponential U \sim \cunit \oplus (U \tensor \exponential U)$.
Furthermore, if $U \oplus (V \tensor W) \leq W$ then $\exponential V \tensor U \leq W$. 
\end{proposition}

A classic consequence of the above is that queries without select can always be denested to a single iteration~\cite{Kozen1997}. However, select breaks denesting since iteration and select do not commute. For instance the following query requires two iterations. The result is that for each of the {f}{i}rst continuation triggered, zero or more instances of the second continuation are triggered. This query can be expressed using sub-queries in the current SPARQL Query working draft~\cite{Seaborne2010}.
\[
\exponential
\cselect{a}
\cselect{n}
\left(
 \left(
  \ask{
   \triple{a}{\textit{name}}{n}
  }
  \then
  P
 \right)
 \tensor
 \exponential
 \cselect{e}
 \left(
  \ask{
   \triple{a}{\textit{email}}{e}
  }
  \then Q
 \right)
\right)
\]

Iteration can be expressed as a colimit of repeated queries. This is a strictly more general property than Proposition~\ref{proposition:exponential}~\cite{Kozen1990}. Since all constructs are colimits which distribute over tensor, the ideals generated by queries form a (commutative) quantale, as exploited by Montanari, Hoare and others~\cite{Montanari1997,Hoare2009}. Quantales are related to spectral theory, which is related to information retrieval techniques used by search engines. {Clarif}{ication} of this connection is future work.

\begin{proposition}
\label{proposition:colimititeration}
Iteration is a colimit of powers of queries. So, $U^n \tensor V \leq W$ for all $n$, if and only if $ \exponential U \tensor V \leq W$.
\end{proposition}

Kozen demonstrates that Boolean algebras can be embedded in Kleene algebras~\cite{Kozen1997}. The `tests' of Kozen correspond to `constraints' in SPARQL. Bisimulation veri{f}{i}es that the Boolean algebra of constraints embeds in the Kleene algebra, in the same manner, with similar consequences.
\begin{proposition}
\label{proposition:Boolean}
The Boolean algebra of constraints embeds in queries. Using standard classical implication, $\phi \Rightarrow \psi$ if and only if $\phi \leq \psi$. Or is choice, and is tensor, exists is select and an iterated constraint is always true.
\begin{comment}
\begin{gather*}
\phi \vee \psi \sim \phi \oplus \psi
\qquad
\phi \wedge \psi \sim \phi \tensor \psi
\qquad
\exists a . \phi \sim \cselect{a} \phi
\qquad
\cunit \sim \exponential \phi
\end{gather*}
\end{comment}
\end{proposition}

As with classical implication, the preorder over triples can be embedded in the partial order over processes. However, since alias assumptions are only a preorder, if $C \sim D$ then it holds that $C \sqsubseteq D$ and $D \sqsubseteq C$, which is weaker than equality. Maintaining distinction of names is important for applications where $\beta$ is not {f}{i}xed over time.

\begin{proposition}
\label{proposition:alias}
$C \sqsubseteq D$ if and only if $C \leq D$.
\end{proposition}

The multiplicatives then, par and times and the units are related in the following manner. Combined with the previous rules the properties of then are established. The second rule shows that `then' can be replaced by the unit delay (as in~\cite{Abramsky1995}).
\begin{proposition}
\label{proposition:continuation}
An empty continuation can be removed, a continuation can be decomposed into the guard and a unit delayed process, and two continuations can be combined in a single par continuation, as follows.
\begin{gather*}
\cunit \then \cend \sim \cunit
\qquad
U \otimes (\cunit \then P) \sim U \then P
\qquad
\left(U \then P \right) \then Q \sim U \then \left(P \cpar Q\right) 
\end{gather*}
\end{proposition}

\begin{comment}
\begin{corollary}
\label{corollary:then}
Some immediate consequences are the following.
\begin{gather*}
(\cunit \then P) \otimes (\cunit \then Q) \sim \cunit \then \left(P \cpar Q\right)
\\
\cscope{a} U \then P \sim \cscope{a} (U \then P) \quad a \not \in \fn{P}
\qquad
(U \oplus V) \then P \sim (U \then P) \oplus (V \then P)
\end{gather*}
\end{corollary}
\end{comment}

\begin{comment}
Every transition can be expressed as a labelled transition where the processes before and after the transition are bisimilar, Theorem~\ref{theorem:completeness-lts}. The proof works by transforming the proof of a transition into a proof of a labelled transition in a normal form.
%
\begin{theorem}[Completeness]
\label{theorem:completeness-lts}
The labelled transition system upto structural congruence is complete with respect to the transition system. That is, if $P \trans Q$ then there exists some $Q'$ such that $P \lts{\cunit} Q'$ and $Q' \equiv Q$.
\end{theorem}

Note however the stronger completeness result that if $P \simeq Q$ the $P \sim Q$ does not hold for all $P$ and $Q$. That is $P \simeq Q$ is not necessarily a bisimulation. 
\end{comment}

The algebra can be applied to optimise queries for distribution. In the example below the {f}{i}rst query is rewritten as the tensor product of two queries.
\[
\begin{array}{l}
\exponential
 \cselect{a}
 \left(
  \left(
   \ask{
    \triple{a}{\textit{knows}}{b_2}
   }
   \then P
  \right)
  \oplus
  \left(
   \ask{
    \triple{a}{\textit{knows}}{b_3}
   }
   \then Q
  \right)
 \right)
%
%\\
%\qquad\qquad\qquad\qquad\qquad\qquad
\sim
\exponential
 \cselect{a}
 \left(
 \cguard{
  \triple{a}{\textit{knows}}{b_2}
 } P
 \right)
\tensor\ 
\exponential
 \cselect{a}
 \left(
 \cguard{
  \triple{a}{\textit{knows}}{b_3}
 } Q
 \right)
\end{array}
\]
The second query above is better for distribution. The tensor product allows two smaller queries to be immediately evaluated in parallel. The tighter scope of the select quanti{f}{i}ers reduces the branching when potential values to select are considered. The distribution of queries across clusters of servers is a major problem for processing Linked Data~\cite{Bizer2009}.

\section{Conclusion}

The calculus introduced provides the {f}{i}rst operational semantics for SPARQL Query -- a W3C recommendation for querying Linked Data. The calculus has a concise logical semantics {def}{ined} by a reduction system. The power of the calculus lies in the synchronisation primitives for queries. The synchronisation primitives are required to match the expressiveness of the core of SPARQL Query. Queries are internalised in a high-level process calculus, where query results determine continuation processes.

An alternative labelled transition system is shown to match the expressive power of the reduction system. Furthermore, the notion of bisimulation in the labelled transition system is sound with respect to equivalence in the reduction system. Bisimulation is used to verify an algebra over queries, which extends existing notions of an algebra for SPARQL Query. An algebra of queries is useful when tackling problems associated with Linked Data, such as distributed query planning.

The operational semantics combines several formalisms, as expected for a real language. The queries form a semiring, which provides a natural partial order. This partial order is used to characterise choice, selects and iteration as colimits. Also, iteration is the least {f}{i}xed point of a monotonic map over queries, hence queries form a Kleene algebra. A preorder over URIs allows small permissible mismatches between content and queries to be resolved, capturing key features of the RDFS standard. Also, a Boolean algebra of constraints is naturally embedded in queries, to provide further control. The calculus demonstrates that key features of SPARQL and related standards for Linked Data can be tightly integrated in one framework.

\bibliographystyle{eptcs}
\bibliography{syndmeta}

\end{document}